 \documentclass[10pt,twocolumn,twoside]{IEEEtran} 
\ifCLASSINFOpdf
\else
\fi

\usepackage{setspace}
\usepackage{algorithm}
\usepackage{algorithmicx}
\usepackage{algpseudocode}
\usepackage{epsfig}  
\usepackage{graphicx}  
\usepackage{longtable}
\usepackage{amsmath}  
\usepackage{array}
\usepackage{multirow}
\usepackage{amssymb}  
\usepackage{lscape}  
\usepackage[justification=raggedright]{caption}	
\newtheorem{theorem}{Theorem}[section]
\newtheorem{lemma}{Lemma}[section]
\newtheorem{example}{Example}[section]

\begin{document}
%
\title{A System-Theoretic Clean Slate Approach to
Provably Secure Ad Hoc Wireless Networking}
%
%
%

\author{Jonathan Ponniah*,~\IEEEmembership{Member,~IEEE,}
        Yih-Chun Hu*,~\IEEEmembership{Member,~IEEE,}
        and~P. R. Kumar**,~\IEEEmembership{Fellow,~IEEE}
\thanks{*CSL \& ECE, Univ. of Illinois, 1308 West Main St., Urbana, IL 61801. Email: \{ponniah1,yihchun\}@illinois.edu. Tel: 217-333-4220.

**Corresponding author: ECE, Texas A\&M University, 3259 TAMU, College Station, TX 77843-3259. Email: prk@tamu.edu. Tel: 979-862-3376.

This paper is partially based on work supported by
NSF under Contract Nos. CNS-1302182, CCF-0939370 and
CNS-1232602, AFOSR under Contract No. FA-9550-13-1-0008, and USARO under Contract No. W911NF-08-1-0238.}}
\maketitle

\begin{abstract}

Traditionally, wireless network protocols have been designed for performance. Subsequently, as attacks have been identified, patches have been developed. This 
has resulted in an ``arms race'' development process
of discovering vulnerabilities and then patching them. 
The fundamental difficulty with this approach is that
other vulnerabilities may still exist. 
No provable security or performance guarantees can ever be provided.

We develop a system-theoretic approach to security that provides a complete
protocol suite with provable guarantees, as well as proof of
min-max optimality with respect to any given utility function of
source-destination rates.
Our approach is based on a model
capturing the essential features of
an ad-hoc wireless network that has been infiltrated with hostile nodes.
We consider any collection
of nodes, some good and some bad, possessing specified capabilities vis-a-vis cryptography, wireless communication and clocks.
The good nodes do not know the bad nodes.
The bad nodes can collaborate perfectly,
and are capable of any disruptive acts ranging from simply jamming to
non-cooperation with the protocols in any manner they please.

The protocol suite caters to the complete life-cycle,
all the way from 
birth of nodes, through all phases of ad hoc
network formation, leading
to an optimized network carrying data reliably.
It provably achieves the min-max of the utility function,
where the max is over all protocol suites published and followed by the
good nodes, while the min is over all Byzantine behaviors of the bad nodes.
Under the protocol suite, the bad nodes do not benefit from any
actions other than jamming or cooperating.

This approach supersedes much previous work that deals with several types of attacks including wormhole, rushing, partial deafness, routing loops, routing black holes, routing gray holes, and network partition attacks.
\end{abstract}



%
\IEEEpeerreviewmaketitle

\section{Introduction} \label{Introduction}
%
%
%
%


\IEEEPARstart{O}{ur} focus is on the problem of security of ad-hoc, multi-hop, wireless networks.
The wireless nodes in these types of networks need to
determine when to transmit packets and at what power levels, discover routes from sources to destinations, and ensure overall end-to-end reliability, all without any centralized controller guiding the process. This requires a suite consisting of multiple protocols.

Several candidate have been proposed.
Medium access control protocols include
IEEE 802.11 \cite{ieee96} and MACAW \cite{bharghavan1994macaw},
power control protocols include COMPOW \cite{KawadiaK05} and
PCMA \cite{monks2001power}, routing protocols include 
DSDV \cite{PerkinsB94}, AODV \cite{AODVRFC}, DSR \cite{DSRRFC}, and OLSR \cite{OLSRRFC}, and transport protocols include TCP \cite{stevens1995tcp} and  variations for ad hoc networks \cite{atcp,adtcp,crvp01,tcp-door}.

All the above protocols are designed on the assumption that all nodes are ``good,''
and will conform to the protocol.
Some nodes can however be malicious, deliberately intent on disrupting the network, a vulnerability especially acute
since the very purpose of ad hoc networks is to allow any node to join a network.
For wireless networks used in safety-critical applications,
e.g., vehicular networks, vulnerabilities can be dangerous.
Moreover, many wireless networking protocols have been
based on wireline protocols, with possible susceptibilities to novel
over the air attacks.

The assumption of benignness, implicit or explicit,
has been the traditional starting point of protocol development.
Systems have been first designed to provide high performance.
Subsequently, as vulnerabilities have been discovered, they have been patched on a case by case basis.
For example, the ``wormhole'' attack was discovered in \cite{HuPerrig2003}, where
an attacker sets up a false link between two nodes.
It is countered by a fix using temporal and geographical packet leashes \cite{HuPerrig2003,Poturalski2008}.
The ``rushing'' attack against DSR was discovered in
\cite{HuPerrigRushing}, in which attackers manipulate the network topology.
This is countered by a fix using network discovery chains.
The ``partial deafness'' attack against 802.11 was discovered in \cite{Choi}, in which an attacker artificially reduces its link quality to draw more network resources.
It is countered by a fix using queue regulation at the access point.  Other attacks against DSR are the routing loop attack in which an attacker generates forged routing packets causing data packets to cycle endlessly; the routing black hole attack in which an attacker simply drops all packets it receives; and the network partition attack in which an attacker injects forged routing packets 
to prevent one set of nodes from reaching another.  These attacks are all countered in the Ariadne protocol \cite{HuAriadne} by the joint use of routing chains, encryption, and packet leashes.  Some protocols such as Watchdog and Pathrater \cite{Marti2000} try to pre-empt attacks by maintaining a blacklist that tracks malicious behavior, but this backfires if an attacker maligns a good node, causing other good nodes to add that node to their blacklists.  
These attacks are not targeted at violating privacy of communications between nodes, which can be avoided simply by encryption. 
Rather, they are generally Denial of Service attacks (DoS), which usually take advantage of algorithms that assume the participating users are good or cooperative.  

The basic problem with this arms race approach of hardening algorithms initially designed for good performance is that
one never knows what other vulnerabilities or attacks exist.
Thus no guarantees can be provided about the security of the protocols at any stage of the arms race process. 
 
Our goal in this paper is to propose an alternate clean slate system-theoretic
approach to security that provides provable performance guarantees. 
We pursue a model-based approach,
comprising a physical model of node capabilities, clocks,
cryptography, and wireless communication.
It is an initial attempt to holistically model
the entire dynamics of an ad-hoc wireless network that has been infiltrated with hostile nodes.

Our goal is to design a protocol suite for the complete life-cycle of the wireless
system, all the way from the very birth of the nodes, 
and continuing through all phases of the network formation process,
to a long-term operation
where the network is carrying data reliably from
sources to their destinations.
The good nodes don't know who the bad nodes are, and are required to follow
the published protocol suite.
Throughput all phases, the bad nodes 
can perfectly collaborate and incessantly indulge in any disruptive
behavior to make the network formation and operation dysfunctional.
They could just ``jam,''
or engage in more intricate behavior
such as not relay a packet,
advertise a wrong hop count, advertise a wrong logical topology, cause packet 
collisions,
disrupt attempts at cooperative scheduling, drop an ACK,
refuse to acknowledge a neighbor's handshake, or
behave inconsistently.


We design
a protocol suite that is provably secure against all such attacks by the malicious
nodes. Not only that, it guarantees min-max optimal performance.
%
%
The
performance is
described by a given utility function, which the good nodes wish to maximize
by publishing a complete
protocol suite and conforming to it.
The bad nodes on the other hand aim to minimize this utility by
indulging in all manner of ``Byzantine''
behavior described above not conforming to the protocol.

This leads to a zero-sum game.
Since the good nodes first announce the protocol,
the best value of the utility function that the good nodes can hope to
attain is its max-min, where the maximization is over all
protocol suites,
and the minimization is over all Byzantine behaviors
of the bad nodes.
We will prove that the protocol suite designed attains this
max-min to within any $\epsilon >0$.
Moreover, we establish three even stronger results.
%

First, this game actually has a saddle point,
i.e., the protocol suite attains the min-max (to within any $\epsilon >0$).
(Generally, min-max results in a higher utility than max-min, since
the bad nodes have to first disclose their tactics).

Second,
the bad nodes can do no better than just jamming or conforming to the published
protocol suite on each ``concurrent transmission vector,''
a generalization of the notion
of an ``independent set'' of nodes that can simultaneously transmit.
They do not benefit from more elaborate Byzantine antics.

Third, the protocol optimally exploits any non-hostile
behavior of the bad nodes.
If they behave suboptimally, i.e., are not as hostile as they could be,
then it will take advantage.
This is a desirable feature since while one wants to
design protocols that are guaranteeably secure in the worst case, one would want
them to exploit any benignness in the environment.

Some important qualifications need to be noted.
First, the results are valid only for the postulated model
of the network. Future research may identify
technological capabilities outside the model
that can attack the protocol suite.
Such discoveries will, one hopes, lead to the development
of more general models and procotols provably secure in them.
The research enterprise will thereby be elevated to a higher level;
instead of reacting to
each proposed protocol one reacts to each proposed model,
with provable guarantees provided at each step.
Section \ref{Concluding Remarks} provides some such directions
for model generalization.

Second, though not merely asymptotic, the optimality is over a large time period, and the overhead of transient phases of the protocol
may be high.
However, there is much scope for optimizing protocol overhead
while preserving security.


Third, how should one view the proposed protocol suite?
The answer is layered.
At a minimum, it can be regarded as a constructive
existence proof that one can indeed provide
optimal performance while guaranteeing security,
with the identified model class only serving as an exemplar of conditions
under which this can be done.
To a more receptive reader, the designed protocol suite is
suggestive of of how one can do so.
The architectural
decomposition into several phases could
perhaps be kept in mind by future protocol designers.

At any rate, one hopes that this approach will trigger several
critical reactions among a skeptical readership, and
lead to follow up work that designs protocols with guaranteed
security and performance for more general model classes.

Section \ref{The Model} describes the model,
Section \ref{The Main Results} the main results, Section \ref{The Outline of the Approach} an outline of the approach,
Section \ref{The Phases of the Protocol Suite} the
protocol suite, and
Section \ref{Feasibility of Protocol and its Optimality} proves
feasibility and optimality.

\section{The Model} \label{The Model}

The model of
an ad-hoc wireless network
infiltrated by hostile nodes
can be organized into four categories: the nodal model {\bf (N)}, communication model {\bf (CO)}, clock behavior {\bf (CL)}, and cryptographic capabilities {\bf (CR)}.

{\bf Nodal model:} {\bf (N1)} There are $n$ nodes, some good and some bad. 
Let $G$ denote the set of good nodes, and its complement $B$ the set of bad nodes.
{\bf (N2)} The good nodes do not know who the bad nodes are a priori.
{\bf (N3)} The bad nodes are able to fully coordinate their actions, and are fully aware of their collective states (equivalent to unlimited bandwidth between them).
{\bf (N4)} The good nodes are all initially powered off, and they all turn on within $U_{0}$ time units of the first good node that turns on. 

{\bf Communication model:}
{\bf (CO1)} Each node $i$ can choose from among a finite set of transmission/reception modes
$M_{i}$ at each time. Each mode corresponds, if transmitting, to
a joint choice of power level, modulation scheme and encoding scheme for each other intended receiver node,
or to just listening and not transmitting, or even to ``jamming,'' which
simply consists of using its power output to emit noise.
{\bf (CO2)} The good nodes are half-duplex, i.e., cannot transmit and receive simultaneously.
{\bf (CO3)} We call $c = ( c_{1}, c_{2}, \ldots , c_{n})$
denoting the mode choices of all the nodes made at a certain time,  as
a ``concurrent transmission vector'' (CTV).
(It is more general than an independent set that is sometimes used
to model wireless networks).
We will denote by $c_{G} = (c_{i}: i \in G )$ and $c_{B} = (c_{i}: i \in B )$
the vectors of choices of modes made by the good and bad nodes respectively,
with each $c_{i} \in {\cal{M}}_{i}$, and let ${\cal{C}}_{G}$ and ${\cal{C}}_{B}$
denote the sets of all such choices. We will denote by
${\cal{C}} := {\cal{C}}_{G} \times {\cal{C}}_{B}$, the set of all CTVs.
{\bf (CO4)} Each $c$ results
in a ``link-rate vector'' $r(c)$ of dimension $n(n-1)$. 
Its $ij$-th component, $r_{ij}(c)$, is the data rate at which
bits can be sent from node $i$ to node $j$ at that time.
Due to the shared nature of the wireless medium,
the rate depends on the transmission mode choices made
by all the other nodes, as well as
the geographic locations of the nodes,
the propagation path loss, the ambient noise, and all other physical characteristics
affecting data rate. 
A component $r_{ij}(c)$ may be zero, for example
if the SINR at $j$ is below a threshold value for decoding, or if node $i$ is not transmitting to node $j$.
{\bf (CO5)} 
If a certain rate vector is achievable then lower rates are also achievable.
To state this,
let $\Lambda := \{ r_{ij} (c) : i \neq j , c \in {\cal{C}} \}$
denote the finite set of
all possible rates than can be achieved.
We suppose that
for every $c$, and $r' \leq r(c)$ (understood component wise)
with all elements in $\Lambda$,
there is a choice $c' \in {\cal{C}}$ such that $r(c') = r'$.
This assumption is not strictly necessary, but it
helps to simplify the statement that bad nodes can claim 
to receive only at low rates.
{\bf (CO6)} In the case of a bad node $j$, the rate $r_{ij}(c)$ may be the result of
some other bad node being able to decode
the packet from $i$ at that rate, and then passing on that
packet to $j$, since bad nodes can collaborate perfectly.
In the case of a bad node $i$, the rate $r_{ij}(c)$ may be the result of
some other bad node being able to transmit
the packet successfully to $j$ at that rate, pretending to be $i$.
Meanwhile, in either case, the bad node may be jamming.
Thus a bad node can both jam and appear to be cooperating, whether
transmitting or receiving, at the same time.
{\bf (CO7)} 
The bad nodes can claim to have received transmissions from each other
at any of the rates in the finite set $\Lambda$, as they please.
To state this, 
for $c = (c_{G}, c_{B})$,
we will partition the resulting link-rate vector as
$r(c) = (r_{GG}(c), r_{GB}(c), r_{BG}(c), r_{BB}(c))$,
where $r_{BG}$ denotes the link-rates from the bad nodes to the good nodes, etc.
We suppose that for every 
$c = (c_{G}, c_{B})$ and every $r'$
with all elements in $\Lambda$,
there is a $c'_{B} \in {\cal{C}}_{B}$ such that
$r(c_{G}, c'_{B}) = (r_{GG}(c), r_{GB}(c), r_{BG}(c), r')$.
{\bf (CO8)} The good nodes know $\Lambda$,
and an upper bound on the cardinalities of the ${\cal{M}}_{i}$'s,
but do not
know the values of the vectors $r(c)$ for any $c \in {\cal{C}}$.
{\bf (CO9)} The assumption that the link-rate vector $r(c)$
does not change with time
implicitly assumes that nodes are not mobile to any significant extent.
We comment further about this assumption in
Section~\ref{Concluding Remarks},
%
%
%

{\bf Clock model:} {\bf (CL1)} Each good node $i$ has a local continuous-time clock that it initializes to zero when it turns on. Its time $\tau^{i}(t)$ is affine with respect to some reference time $t\geq0$, i.e.,  $\tau^{i}(t)=a_{i}t+b_{i}$ where $a_{i}$ and $b_{i}$ are called the skew and offset respectively.  
Wlog, the time $t$ above and in (N4) is
taken equal to the clock time of the first good node to turn on.  
{\bf (CL3)} Denoting the relative skew and offset between nodes $i$ and $j$ by $a_{ij}:=\frac{a_{i}}{a_{j}}$ and $b_{ij}:=b_{i}-a_{ij}b_{j}$,  
node $i$'s time with respect to node $j$'s time $s$ is $\tau^{i}_{j}(s) =a_{ij}s+b_{ij}$.
We assume $0<a_{ij}\leq a_{max}$.  As a corollary of (N4,CL1,CL3), $|b_{ij}|\leq a_{max}U_{0}$, since $\tau^{i}(U_{0})\geq0$.  {\bf (CL4)} The good nodes do not know their skew or offset a priori.  {\bf (CL5)} Finally, due to its digital processor, a good node $i$ can only observe a quantized version of its continuous-time local clock $\tau^{i}(t)$. 

{\bf Cryptographic capabilities:} {\bf (CR1)} Each node is assigned a
public key and a private key; information encrypted by a private key
can only be decrypted with the corresponding public key.
The private key is never revealed by a good node to any other node.
Possession of a public key does not enable an attacker to forge,
alter, or tamper with an encrypted packet generated with the
corresponding private key.  The good nodes
encrypt all their transmissions. 
{\bf (CR2)} Each node possesses the public key of a central authority.
{\bf (CR3)} Each node possesses an identity certificate, signed by the
central authority, containing node $i$'s public key and ID number.
The certificate binds node $i$'s public key to its identity.  {\bf
(CR4)} Each node possesses a list of all the other $n$ node IDs.


\section{The Main Results} \label{The Main Results}

Each time that the good nodes make a certain choice
$c_{G}$, the bad nodes
could respond with some choice drawn only from
a certain subset ${\cal{C}}_{B,c_{G}} \subseteq
{\cal{C}}_{B}$.
In this way they could ensure that only the subset ${\cal{E}} := \{ (c_{G},c_{B}):
c_{G} \in {\cal{C}}_{G}, c_{b} \in {\cal{C}}_{B,c_{G}} \}$ is ever employed by the network.
If so, we will say that ${\cal{E}}$ is \emph{enabled}, while its complement
${\cal{D}} := {\cal{C}} \setminus
\{ (c_{G},c_{B}):
c_{G} \in {\cal{C}}_{G}, c_{b} \in {\cal{C}}_{B,c_{G}} \}$ is \emph{disabled} by the bad nodes.
We will denote by $\Delta$ the set of all such sets ${\cal{D}}$ that they have the capability to disable.
For any set ${\cal{E}}$ of enabled CTVs,
let
${{R(\cal{E})}} :=$ ConvexHull$( \{ r(c) : c \in {\cal{E}} \})$
\emph{be the set of link rate-vectors supported by} ${\cal{E}}$, i.e., generated by time sharing over ${\cal{E}}$.
Let  ${\cal{G}}( {{\cal E}})$ be a directed graph
over the nodes, where there is an edge
$ij$
if and only if $r_{ij} (c)>0$ for some $c \in {{ R}}( {{\cal E}})$.

We assume that the good nodes can communicate in a multi-hop fashion with each other over bidirectional links at some minimal positive rate, regardless of what the bad nodes do: \\
\noindent
{\bf Connectedness Assumption (C):}
Let ${\cal{G}}^{*} := {\cal{ G}}( {{\cal C}} \setminus {{\cal D}}^{*})$ be the
graph resulting from the  maximum set ${\cal{D}}^{*} \in \Delta$ that the
bad nodes can disable.
We will assume that the good nodes are connected
in the subgraph of ${\cal{G}}^{*}$ that consists only of
edges $ij$ for which both $ij$ as well as $ji$ are edges in ${\cal{G}}^{*}$.

Denoting by $P_{ij}$ the set of all paths from $i$ to $j$,
the \emph{multi-hop capacity region} 
of  $n(n-1)$-dimensional end-to-end source-destination throughput vectors is
defined in the standard way as
${{ C}}( {\cal{ E}}) := \{x:$ For some vector $y \geq 0$
with $0 \leq \sum_{p: \ell \in p} y_{p} \leq r_{\ell}$
for some $r \in {{R(\cal{E})}}$,
$x_{ij}= \sum_{p \in P_{ij}} y_{p}$ for all $1 \leq i,j \leq n, j \neq i \}$.

We employ a utility function defined over 
the throughputs of any subset of
source-destination pairs of interest: \\
{\bf Utility function assumption (U):}
For any subset $S \subseteq \{ 1, 2, \ldots , n \}$
and any throughput vector $x$, let $U(x,S )$ depend
only on $x_{ij}$ for $i,j \in S$.
For every $S$, $U(x, S )$ is continuous and monotone increasing
in the components of  $x$.

We now consider the game where 
the good nodes wish to maximize it for the nodes
perceived to be good, while the bad nodes wish to minimize it over all their Byzantine behaviors.  
To obtain an upper bound on utility, 
suppose that the bad nodes disable only the CTVs in
${\cal{D}}$ and
reveal this choice to the good nodes.
Let ${\cal{E}} := {\cal{C}} \setminus {\cal{D}}$.
If ${\cal{G}} ({\cal{E}})$ has several strongly connected components, then,
by the connectedness assumption (C), the good nodes are all in the same
component, denoted by $F({\cal{E}})$, and thus know that
the nodes outside $F({\cal{E}})$ are bad.
They will therefore only consider the utility accrued as $U(x,F({\cal{E}}))$,
and maximize it over all $x \in {\cal{C}} ({\cal{E}})$.
Hence an upper bound on achievable utility is $\displaystyle\min_{\substack{{{\cal{D}} \in \Delta}}}\hspace{2mm}\displaystyle\max_{\substack{x \in {\cal{C}}({\cal{C}} \setminus {\cal{D}})}}U(x,F({\cal{C}} \setminus {\cal{D}}))$. Our main result, elaborated on in Theorem \ref{main}, is:

\vspace{0.05in}

\begin{theorem}
Consider a network that satisfies (N), (CO), (CL), (CR), (C) and (U).  Given an arbitrary $\epsilon$, where $0<\epsilon<1$, the protocol described in Section \ref{sec:protocol} ensures that all the good nodes
obtain a common estimate of the component that they are all members of, and achieves the utility 
\begin{align}
	\label{upperbound}
	(1-\epsilon) \displaystyle\min_{\substack{{{\cal{D}} \in \Delta}}}\hspace{2mm}\displaystyle\max_{\substack{x \in {\cal{C}}({\cal{C}} \setminus {\cal{D}})}}
	U(x,F({\cal{C}} \setminus {\cal{D}})).
\end{align}
\end{theorem}
\vspace{0.05in}

%

Some important consequences are the following.
Normally, one would expect that since the good nodes have to first declare their protocol and follow it, they can only attain ``max-min,'' which is generally smaller than min-max. Since the latter can be attained (arbitrarily closely), it
shows firstly that the bad nodes are unable to benefit from having a priori knowledge of the protocol.  Second, since all that the bad nodes can benefit from is deciding
which sets to disable, they are effectively limited to jamming and/or cooperating in each CTV.  
Other more Byzantine behaviors are not any more effective.

The example below shows why a bad node may prefer to ``conform'' rather than jam
for some utility functions. \\
\begin{example} \label{example}
Consider the network of Figure 1.
Nodes 1 and 2 are good and in close proximity,
while node 3 is bad and located far away.  Consider the ``fairness-based ''utility function 
$U(x) := \min \{ x_{12}, x_{32} \}$.  If node 3 jams, then
the connected component becomes $\{1,2\}$, and the good nodes
proceed to maximize only $x_{12}$, which node 3 can only slightly
impinge because it is so far away
from node 2.
However, if node 3 cooperates, then the connected component
is $\{1,2,3\}$, and the optimal solution for this
``fair'' utility function is to make $x_{32} = x_{12}$. However, link 32 
being weak, it requires much more airtime than link 12, thus considerably reducing $x_{12}$.
\end{example}
\begin{figure}[H] \label{fig-example}
\begin{center}
	\includegraphics[width=5cm]{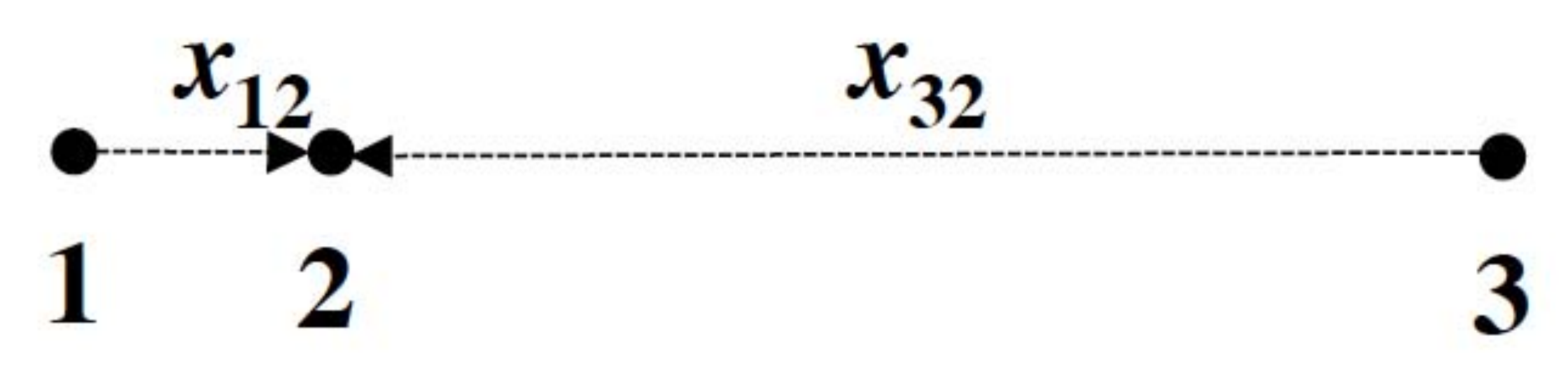}\\	

\caption{Example \ref{example}.}
\end{center}
\end{figure}

\section{The Outline of the Approach} \label{The Outline of the Approach} 
The heart of the approach is to investigate different CTVs,
exploiting the fact that the operation of the network consists of invoking
which such set to use at any given instant.
If a good node fails to receive a scheduled packet transmitted during a CTV set, then that good node alerts the rest of the network during a verification phase, and the offending CTV set is never used again.
After each such pruning the network then re-optimizes its utility over the remaining
CTVs. The decreasing sequence of remaining
sets of CTVs necessarily converges to an operational collection
of CTVs,
over which the utility is optimized by time sharing.
Since the set of disabled CTVs is determinable
by the network, as we show, it is the same as if it were revealed to the good nodes a priori,
which allows achievement of min-max.
It also shows why more complex Byzantine
behaviors than jamming or cooperating
are not any more effective for the bad nodes.

There are however several problems that lie along the way to realizing this scheme.
First, all of the above presumes that all the nodes are good, and,
second, also that the nodes know the network topology and other parameters,
both of which are false.
This leads to the challenge: How to determine the network, while under attack from bad nodes
when one does not know the network?
We present a complete protocol suite that proceeds through several phases
to achieve this end result.

After their birth, the nodes need to first discover who their neighbors are.
This requires a two-way handshake, which presents one problem already.
Two good nodes that are neighbors
can successfully send packets to each other if there are no primary (half-duplex)
or secondary (collision) conflicts.
To achieve this we employ an Orthogonal MAC Code
\cite{ponhukum12}.
Next, the two nodes need to update their clock parameters.
After this, the nodes propagate their neighborhood information
so that everyone learns about the network topology.
This also poses some challenges when there are intermediary
bad nodes. This is addressed by a
version of the Byzantine General's algorithm of \cite{BarNoy1987},
by capitalizing on connectedness assumption (C).
Next, even though all the good nodes converge to a common network view, that view may be internally inconsistent, especially with respect to clocks.
To resolve this we employ a certain consistency check algorithm.
Next, the nodes proceed to determine an optimal schedule for
time sharing over the set of CTVs that have performed
consistently from the very beginning, 
and execute it. However, a bad node that has cooperated hitherto
may not cooperate at  this point. Hence the results of this operational phase need
to be verified,
the dysfunctional
CTV pruned, the schedule re-optimized, and the procedure iterated.

The reader may wonder: Why do we even need a notion of ``time''? 
First, without it, we cannot even speak of throughput or thus of
utility.
Second, 
we use local clocks to schedule transmissions and coordinate activity
(as is  
quite 
common, e.g., time-outs in MAC and transport protocols).
On the other hand, dependence on distributed synchronized clocks for coordinated activity opens yet another avenue for bad nodes to sabotage the protocol -- interfering with  the 
clock synchronization algorithm.
Therefore, topics like scheduling, clock synchronization, 
utility maximization, and security, are deeply interwoven.
Therefore one needs a holistic approach that
addresses all these issues at every stage of the operating
lifetime,
and guarantees overall security and min-max optimality.
This is the raison d'$\rm{\hat{e}}$tre for this paper.

\section{The Phases of the Protocol Suite} \label{The Phases of the Protocol Suite} \label{sec:protocol}
The protocol suite consists of six phases: Neighbor Discovery,
Network Discovery, Consistency Check, Scheduling, Data Transfer, and Verification.
Proofs are deferred to Section \ref{Feasibility}.

We first note the necessity for
a key ingredient.
Even two good nodes that are neighbors as in assumption (C)
are only guaranteed to be able to
successfully send packets to each other
provided one is transmitting,
the other is listening (since good nodes are half-duplex), and the remaining good nodes
are all silent.
The Orthogonal MAC Code (OMC) of \cite{ponhukum12} ensures 
the simultaneity of all these events, even though the clocks
of different nodes have different skews and offsets.
For each pair of nodes $i,j$, it defines certain zero-one valued functions of local
time at each node, such that if $i$ transmits a packet of duration $W$ to $j$
at that time, then the
packet is successfully received, and the delay involved
in waiting for such an eventuality is never more than a certain
$T_{MAC}(W)$.

\subsection{The Neighbor Discovery Phase}

In this phase, each node $i$ will determine the identity and relative clock parameters of  nodes in its neighborhood ${\cal N}_{i}$, and include this data in a mutually authenticated link certificate.  


In the first two steps, each node $i$ attempts a handshake with a neighbor node $j$ by broadcasting a probe packet $PRB_{ij}$ and waiting for an acknowledgement $ACK_{ji}$.  The probe packet contains an identity certificate signed by a central authority.  Given ${\cal N}_{i}:=\{1,\ldots,n\}$\textbackslash$
i$, 
an initial candidate for the set of bidirectional neighbors of $i$ (as in (C)),
to indicate that node $i$ transmits $PRB_{ij}$ to each node $j \in {\cal N}_{i}$ via the OMC, and receives $PRB_{jj}$ from each node $j\in {\cal N}_{i}$,
we use $\text{TxRxMAC}(PRB_{i\rightarrow{\cal N}_{i}}$,$PRB_{{\cal N}_{i}\rightarrow i})$.  If a probe packet is not received from some node $j$, then $j$ is pruned from ${\cal N}_{i}$.

Next, node $i$ transmits an acknowledgment $ACK_{ij}$ to node $j$ containing a signed confirmation of the received probe packet $PRB_{j}$.  Node $i$ also listens for an acknowledgment $ACK_{ji}$ from node $j$.  Node $i$ further removes from ${\cal N}_{i}$ any nodes that failed to return acknowledgements.   

Then node $i$ transmits to each node $j\in {\cal N}_{i}$ a pair of timing packets $TIM^{(1)}_{i,j}$ and $TIM^{(2)}_{j,i}$ that contain the send-times $s^{(1)}_{ij}$ and $s^{(2)}_{ij}$ respectively as recorded by its local clock $\tau^{j}(t)$.  Node $i$ also receives a corresponding pair of timing packets $TIM^{(1)}_{i,j}$ and $TIM^{(2)}_{j,i}$ from node $j$, and records the corresponding receive-times $r^{(1)}_{ji}$ and $r^{(2)}_{ji}$ respectively, as measured by the local clock $\tau^{i}(t)$.  Any node that fails to deliver timing packets to node $i$ is further removed from ${\cal N}_{i}$.  The timing packets are used to estimate the relative skew $a_{ji}$ 
by $\hat{a}_{ji}:=\frac{r^{(2)}_{ji}-r^{(1)}_{ji}}{s^{(2)}_{ji}-s^{(1)}_{ji}}$.

The relative skew is used at the end of the Network Discovery Phase, to estimate a reference clock with respect to the local continuous-time clock.  In the last two steps, node $i$ creates a link certificate $LNK^{(1)}_{ij}$ containing the computed relative clock skew with respect to node $j$, and transmits this link to node $j$ using the OMC.  Node $i$ also listens for a similar link certificate $LNK^{(2)}_{ji}$ from node $j$.   
	
Finally, node $i$ verifies and signs the received link certificate, and transmits the authenticated version $LNK^{(2)}_{ij}$ back to node $j$.  Node $i$ listens for a similar authenticated link certificate $LNK^{(2)}_{ji}$ from $j$.  Any nodes that fail to return link certificates are removed from the set ${\cal N}_{i}$.  This set now represents the nodes in the neighborhood of node $i$ with whom node $i$ has established mutually authenticated link certificates.  
The Neighbor Discovery Phase's pseudocode is shown in Algorithm \ref{NeighborDiscoveryAlg}.

One problem is that the algorithm must be completed in a partially coordinated manner even though the nodes are asynchronous; the completion of any stage in the Exponential Information Gathering (EIG) algorithm (see below) depends on the successful completion of the previous stages by all other good nodes.  Consequently, we assign increasingly larger intervals $S_{k}:=[t_{k},t_{k+1}),k=1,\ldots6,$ to each successive protocol stage; see Section~\ref{Feasibility of Protocol and its Optimality}.        

\begin{algorithm} 
\begin{scriptsize}
	\caption{The Neighbor Discovery Phase}
	\label{NeighborDiscoveryAlg}
		\begin{algorithmic}
			\Procedure{NeighborDiscovery}{}
			\State ${\cal N}_{i}:=\{1,\ldots,n\}\setminus i$
			\While{$t\in S_{1}$}
			\State \Call{TxRxMAC}{$PRB_{i\rightarrow{\cal N}_{i}}$,$PRB_{{\cal N}_{i}\rightarrow i}$}
				\State \Call{Update}{${\cal N}_{i}$}
			\EndWhile
			\While{$t\in S_{2}$}
			\State \Call{TxRxMAC}{$ACK_{i\rightarrow{\cal N}_{i}}$,$ACK_{{\cal N}_{i}\rightarrow i}$}
							\EndWhile
			\While{$t\in S_{3}$}
			\State \Call{TxRxMAC}{$TIM^{(1)}_{i\rightarrow{\cal N}_{i}}$,$TIM^{(1)}_{{\cal N}_{i}\rightarrow i}$}
				\State \Call{Update}{${\cal N}_{i}$}
			\EndWhile
			\While{$t\in S_{4}$}
			\State \Call{TxRxMAC}{$TIM^{(2)}_{i\rightarrow{\cal N}_{i}}$,$TIM^{(2)}_{{\cal N}_{i}\rightarrow i}$}
				\State \Call{Update}{${\cal N}_{i}$}
			\EndWhile
			\While{$t\in S_{5}$}
			\State \Call{TxRxMAC}{$LNK^{(1)}_{i\rightarrow{\cal N}_{i}}$,$LNK^{(1)}_{{\cal N}_{i}\rightarrow i}$}
				\State \Call{Update}{${\cal N}_{i}$}
			\EndWhile
			\While{$t\in S_{6}$}
			\State \Call{TxRxMAC}{$LNK^{(2)}_{i\rightarrow{\cal N}_{i}}$,$LNK^{(2)}_{{\cal N}_{i}\rightarrow i}$}
				\State \Call{Update}{${\cal N}_{i}$}
			\EndWhile
			\EndProcedure
		\end{algorithmic}
\end{scriptsize}
\end{algorithm}
\subsection{The Network Discovery Phase}

The purpose of this Phase is to allow the good nodes to obtain a \emph{common} view of the network topology and \emph{consistent} estimates of all clock parameters.  To accomplish this, the good nodes must disseminate their lists of neighbors to all nodes, so that all can decide on the same topology view.  However
good nodes do not know a priori which nodes are bad, and  so bad nodes can selectively drop lists or introduce false lists to prevent consensus.  We resolve this by using a version of the Byzantine General's algorithm of \cite{BarNoy1987},
requiring an EIG tree data structure.  Let $T_{i}$ denote node $i$'s EIG tree, which by construction has depth $n$.  The root of $T_{i}$ is labelled with node $i$'s neighborhood, i.e., the nodes in ${\cal N}_{i}$ and the corresponding collection of link certificates.  First node $i$ transmits to every node $j \in {\cal N}_{j}$ in its neighborhood, the list of nodes in ${\cal N}_{i}$ and corresponding link certificates, while receiving similar lists from each node in ${\cal N}_{j}$.  Node $i$ updates its EIG tree with the newly received lists from its neighbors, by assigning each received list to a unique child vertex of the root of $T_{i}$.  Node $i$ then transmits the set of level 1 vertices of $T_{i}$ to every node in its neighborhood, receiving a set of level 1 vertcies from each neighbor in turn.  The EIG tree $T_{i}$ is updated again.  This process continues through all $n$ levels of the EIG tree.  

The notation $T^{(k)}_{i}$ in Algorithm~\ref{EIGByzMAC} indicates the $k$-level vertices of the EIG tree $T_{i}$.  The notation $\text{TxRxMAC}(T^{(k)}_{i\rightarrow{\cal N}_{i}},T^{(k)}_{{\cal N}_{i}\rightarrow i})$ indicates that, using the OMC, node $i$ transmits $T^{(k)}_{i}$ to each node $j \in {\cal N}_{i}$, and receives $T^{(k)}_{j}$ from each node $j\in {\cal N}_{i}$.

We use $\text{UPDATE}(T_{i})$, to update the EIG tree $T_{i}$ after the arrival of new information, and the procedure $\text{DECIDE}(T_{i})$ to infer the network topology based on the EIG tree.  The $n$-stage EIG algorithm guarantees that if the subgraph of good nodes is connected, then each good node will decide on the same topological view.
\begin{algorithm}
\begin{scriptsize}
	\caption{The EIG Byzantine General's Algorithm}
	\label{EIGByzMAC}
	\begin{algorithmic}
		\Procedure{EIGByzMAC}{${\cal N}_{i}$}
			\State $T^{(0)}_{i}:={\cal N}_{i}$
			\For{$k=1,\ldots n$}
				\While{$t\in S_{6+k}$}
				\State \Call{TxRxMAC}{$T^{(k)}_{i\rightarrow{\cal N}_{i}},T^{(k)}_{{\cal N}_{i}\rightarrow i}$}
					\State \Call{Update}{$T_{i}$}
				\EndWhile
			\EndFor
			\State \Call{Decide}{$T_{i}$}
		\EndProcedure
	\end{algorithmic}
\end{scriptsize}
\end{algorithm}

\subsection{The Consistency Check Phase} \label{The Consistency Check Phase}
Unfortunately, a fundamental difficulty is that malicious nodes along a path $1,\ldots,n$ may have generated false time stamps in the Neighbor Discovery Phase, and thus corrupted the measured relative skews between adjacent nodes.  
There may be several connecting paths infiltrated by bad nodes that thereby generate different values for the relative skew.  It is impossible to determine the correct path from the relative skews alone.  
Every pair of such inconsistent paths 
corresponds to an inconsistent cycle in which the skew product is not equal to one.  We use an algorithm called Consistency Check to identify the path that generated the correct relative skew.  

Consistency Check works by circling a timing packet around every cycle in which the skew product differs from one by more than $\epsilon_{a}$,
a desired maximum skew error.
At the conclusion of the test, at least one link with a malicious endpoint will be removed from the cycle, eliminating a connecting path.  During the test, each node in such a cycle is obliged to append a receive time-stamp and a send time-stamp generated by the local clock before forwarding the packet to the next node.  These time-stamps must satisfy a delay bound condition; the send time and receive time cannot differ by more than 1 clock count.  A node fails the consistency check otherwise,
or if its time stamps do not agree with its declared relative skew.  The key idea is that if the test starts after a sufficiently large amount of time has elapsed, the clock estimates based on faulty relative skews will have diverged so extensively from the actual clocks that at least one malicious node in the cycle will find it impossible to generate time-stamps that are consistent with its declared relative clock skew and satisfy the delay bound condition (all proofs are in Section \ref{Feasibility}):.  
\begin{theorem}
	\label{bigtheorem}
	Let $T_{j}$ be the start-time of the Consistency Check for the $jth$ inconsistent cycle, consisting of nodes $i_{1},\ldots,i_{m}$.  At least one malicious node in cycle $j$ will violate a consistency check condition, if $T_{j}>\frac{\hat{a}_{i_{m},i^{*}}(m+1)K+\epsilon_{b}}{\epsilon_{a}}$ where $i^{*}$ is the node with the smallest skew product $\hat{a}_{i^{*},i_{1}}$.
\end{theorem}

Algorithm \ref{ConsistencyCheck} depicts Consistency Check.  Given a cycle, $j$, $k$ and $m$ denote nodes that follow and precede node $i$ respectively in the cycle.  If node $i$ is the leader of the cycle, i..e., the node with smallest ID, then node $i$ initiates the timing packet that traverses the cycle and transmits it to node $k$.  Otherwise, node $i$ waits for the timing packet to arrive from node $m$ before forwarding it to node $k$. 

\begin{algorithm}
\begin{scriptsize}
	\caption{Consistency Check Algorithm at Node $i$}
	\label{ConsistencyCheck}
	\begin{algorithmic}
		\Procedure {ConsistencyCheck}{}
		\State $START:=\frac{(n+1)(a_{max})^{n+1}+(n+1)(a_{max})^{n+1}U_{0}}{\epsilon_{a}}$
		\For{\text{each cycle} $C_{j}$}
			\State $k=$\Call{Next}{$C_{j}$}
			\State $m=$\Call{Prev}{$C_{j}$}
			\If{$i$=\Call{Leader}{$C_{j}$} and $t\geq START$}
				\State \Call{Transmit}{$TIM_{i\rightarrow k}$}
			\ElsIf{$i\in C_{j}$}
				\State \Call{Receive}{$TIM_{m\rightarrow i}$}
				\State \Call{Transmit}{$TIM_{i\rightarrow k}$}
			\EndIf
		\EndFor
		\EndProcedure
	\end{algorithmic}
\end{scriptsize}
\end{algorithm}

After all inconsistent cycles have been tested, each node $i$ disseminates the set of all timing packets ${\cal T}_{i}$ it received to other nodes.  The EIG algorithm 
is used to ensure a common view of the timing packets generated.  Each node removes from the topology any link whose endpoints generate time-stamps inconsistent with its declared relative skew or violated the delay bound.  The complete Phase is shown in Algorithm \ref{NetworkDiscovery}.
\begin{algorithm}
\begin{scriptsize}
	\caption{The Network Discovery Phase at Node $i$}
	\label{NetworkDiscovery}
	\begin{algorithmic}
		\Procedure{NetworkDiscovery}{}
		\State \Call{EIGByzMAC}{${\cal N}_{i}$}
			\State \Call{ConsistencyCheck}{}
			\State \Call{EIGByzMAC}{${\cal T}_{i}$}
		\EndProcedure
	\end{algorithmic}
\end{scriptsize}
\end{algorithm}

At the conclusion of Network Discovery Phase node $i$ shares a common view of the network topology with all other good nodes.  As a result, the network can designate the node with smallest ID as the \emph{reference clock}.  Furthermore, each node $i$ has an estimate of the reference clock $\tau^{r}_{i}(t)$ with respect to its local clock $t$ using the formula $\hat{\tau}^{r}_{i}(t):=\hat{a}_{ri}t$, where estimated $\hat{a}_{ri}$ and  actual relative skews $a_{ri}$ differ by at most $\epsilon_{a}$.         

\subsection{The Scheduling Phase}

In the Scheduling Phase the good nodes in the network obtain a common schedule governing the transmission and reception of data packets.  
A ``schedule'' is simply a sequence of CTVs, each with specified start and end times.  Each node $i$ 
divides the Data Transfer Phase into time-slots, and assigns a CTV to each time-slot so that the resulting throughput vector is utility optimal.
All the good nodes independently arrive at the same schedule since they independently optimize the same utility function over the same ${\cal C}$ (ties broken lexicographically).  

Since the good nodes must conform to a common schedule,
 each node $i$ generates a local estimate of the reference clock $\hat{\tau}^{r}_{i}(t)$ with respect to its local clock $t$, as described in the Network Discovery Phase. 
However, this estimate may not be perfectly accurate; 
some of the nodes on a path along which relative skew is estimated
may be malicious and can introduce an error of at most $\epsilon_{a}$ into the computed relative skew.  To address this, the time-slots are separated by a dead-time of size $D$, where given any pair of nodes $(i,j)$, $D$ is chosen to 
satisfy
$|\hat{\tau}^{r}_{i}(t)-\hat{\tau}^{r}_{j}(\tau^{j}_{i}(t))|\leq D$.  

Finally, $n^{2}(n-1)$ time-slots are enough to guarantee that every pair of nodes can communicate once in either direction, via multihop routing, during Data Transfer Phase.  
The algorithm $\text{UtilityMaximization}({\cal C})$ for the Scheduling Phase is depicted in Algorithm \ref{Scheduling}.            
At the end
 of Scheduling Phase, node $i$ shares a common utility maximizing schedule with other good nodes.

\begin{algorithm}
\begin{scriptsize}
	\caption{The Scheduling Phase at Node $i$}
	\label{Scheduling}
	\begin{algorithmic}
		\Procedure{Scheduling}{}
		\State \Call{UtilityMaximization}{${\cal C}$}
		\EndProcedure
	\end{algorithmic}
\end{scriptsize}
\end{algorithm}

\subsection{The Data Transfer Phase}

In this Phase the nodes exchange data packets using the generated schedule.  It is divided into time-slots, with each assigned a CTV, a rate vector, and set of packets for each transmitter in the set.  
To prevent collisions resulting from two nodes assigning themselves
to different time slots due to timing error, node $i$ begins transmission $D$ time-units after the start of the time-slot.  The transmitted packet is then guaranteed to arrive at the receiver in the same time slot, for
appropriate choice of $D$ and  time-slot size $B_{slot}$.

Algorithm 7 defines this phase, with $m_{k}$ denoting a message to be transmitted or received by node $i$ in the $k_{th}$ slot, $T_{start}$ the start time of the phase measured by the local estimate of the reference clock $\hat{\tau}^{r}_{i}(t))$, $S_{k}=[t_{k},t_{k+1}),k=1,\ldots,N$ the time-slots of the phase with $N=n^{2}(n-1)$, $t_{1}=T_{start}$, and $t_{k+1}:=t_{k}+B_{slot}+2D$, and  $TX(k)$ and $RX(k)$  the CTV, and receiving nodes during slot $k$.      

\begin{algorithm} \label{Data Transfer Algorithm}
\begin{scriptsize}
	\caption{The Data Transfer Phase at Node $i$}
	\begin{algorithmic}
		\Procedure{DataTransfer}{$T_{start}$}
		\For{k=1,\ldots,N}
			\If{$t\in S_{k}$ and $t\geq t_{k}+D$ and $i\in TX(k)$}
				\State\Call{Transmit}{$m_{k}$}
			\ElsIf{$t\in S_{k}$ and $i\in RX(k)$}
				\State\Call{Receive}{$m_{k}$}
			\EndIf
		\EndFor
		\EndProcedure
	\end{algorithmic}
\end{scriptsize}
\end{algorithm}

\subsection{The Verification Phase}

However, malicious nodes may not cooperate in the Data Transfer Phase.
So whenever a scheduled packet fails to arrive at node $j$, it adds the offending CTV and associated packet number to a list, and disseminates the list in the Verification Phase
using the EIG Byzantine General's algorithm.
These CTVs are then permanently further pruned from the collection of feasible CTVs.  With $L_{k}$ denoting the list that failed during the $k$th iteration of the Data Transfer Phase,
the set ${\cal C}_{k}$ of feasible CTVs during the $k$th iteration of the Scheduling Phase is updated to ${\cal C}_{k+1} = {\cal C}_{k} \setminus L_{k}  $
in Algorithm \ref{Verification}.  

All communication can be scheduled into slots separated by a dead-time of $2D$.  Within each of the $n$ stages of the EIG Byzantine General's algorithm, there are $n(n-1)$ pairs of nodes that may communicate, and at most $n$ nodes on the connecting path.  Therefore, the total number of time slots required is $n^{3}(n-1)$.


At the conclusion of the phase, the good nodes again share a common view of the set of feasible CTVs for the next iteration of the Scheduling Phase.  
\begin{algorithm}
\begin{scriptsize}
	\caption{The Verification Phase at Node $i$}
	\label{Verification}
	\begin{algorithmic}
	  \Procedure{Verification}{}
		\State \Call{EIGByz}{$L_{k}$}
		\State \Call{Update}{${\cal C}_{k+1}$} 
		\EndProcedure
	\end{algorithmic}
\end{scriptsize}
\end{algorithm}

\subsection{The Steady State}
The network cycles through Scheduling, Data Transfer, and Verification Phases for $n_{iter}$ iterations.  
    Eventually, by finiteness, it converges to a set of CTVs, and a utility-maximizing schedule over it.
     The overall protocol is in Algorithm \ref{protocol}.     

\begin{algorithm}
\begin{scriptsize}
	    \caption{The Complete Protocol}
	    \label{protocol}
	    \begin{algorithmic}
		    \State \Call{NeighborDiscovery}{}
		    \State \Call{NetworkDiscovery}{}
		    \For{$k=1,\ldots,n_{iter}$}
		    	\State \Call{Scheduling}{${\cal C}_{k}$}
		    	\State \Call{DataTransfer}{t}
		    	\State \Call{Verification}{}
		    \EndFor
	    \end{algorithmic}
\end{scriptsize}
    \end{algorithm}

\section{Feasibility of Protocol and Optimality Proof} \label{Feasibility of Protocol and its Optimality}
    \label{Feasibility}
   For the distributed wireless nodes to exchange data over the network,
they must not only have the same topological view, in order to independently arrive at a common schedule, but they must also have a consistent view of a reference clock so that any activity will conform to this common schedule.  
 For this, we consider the consistency check algorithm
of Section~\ref{The Consistency Check Phase}.

Consider a chain network $1,\ldots,n$, where the endpoints, nodes $1$ and  $n$ are good, and the intermediate nodes $2,\ldots,n-1$ are bad.  Note that this network can also be reduced to a cycle of size $n-1$ by making both endpoints the same node.  We  assume that the two good endpoints do not know if any of the intermediate nodes are bad.    

Now suppose that each pair of adjacent nodes $(i,i-1)$ for $i=2,\ldots,n$ has declared a set of relative skews and offsets $\{\hat{a}_{i,i-1},\hat{b}_{i,i-1}\}$, and that each node in the chain knows this set.  The two good nodes wish to determine whether the declared skews are accurate, i.e., whether $a_{n,1}=\prod^{n}_{i=2}\hat{a}_{i,i-1}$. Unfortunately, the good nodes have no way of directly measuring $a_{n,1}$.  The estimate of $a_{n,1}$ is obtained from the skew product itself, which is the very quantity that needs to be verified.  So, instead, the good nodes carry out the consistency check described earlier.  After waiting a sufficiently long time, node $1$ initiates a timing packet that traverses the chain from left to right.  Each node in the chain is obligated to forward the packet after appending receive and time-stamps that satisfy the skew consistency and delay bound conditions.

In order to defeat this test, the bad nodes, having collectively declared a false set of relative skews and offsets, must support two sets of clocks for each node $i\in\{2,\ldots,n\}$: a ``left'' clock $\tau^{i,l}(t)$ to generate receive time-stamps, and a ``right'' clock $\tau^{i,r}(t)$ to generate send-time stamps.  Unlike the clocks of the good nodes, the left and right clocks of the bad nodes need not be affine with respect to the global reference clock.  In fact, the bad nodes are free to jointly select any set of clocks $\{\tau^{i,l}(t),\tau^{i,r}(t),\forall i=2,\ldots,n-1\}$ that are arbitrary functions of $t$, a much larger set than the affine clocks being emulated.  However, we will show that if node $1$ waits sufficiently long enough, there is no set of clocks $\{\tau^{i,l}(t),\tau^{i,r}(t),i=2,\ldots,n-1\}$ that can generate time-stamps which satisfy both conditions of the consistency check.

Let $r_{i,i-1}$ and $s_{i,i+1}$ denote the receive and send time-stamps generated by a bad node $i$ with respect to the left and right clocks $\tau^{i,l}(t)$ and $\tau^{i,r}(t)$ respectively.  Let $t_{i,l}$ and $t_{i,r}$ denote the time with respect to the global reference clock at which the receive and send time-stamps are generated at node $i$.  We have $r_{i-1,i} :=\tau^{i,l}(t_{i,l})$ and $s_{i,i+1} :=\tau^{i,r}(t_{i,r})$.
Let $t_{1}$ and $t_{n}$ denote the time with respect to the global reference clock at which the timing packet was transmitted by node $1$ and received by node $n$ respectively.  We have $s_{1,2} :=\tau^{1}(t_{1}), _{n-1,n} :=\tau^{n}(t_{n})$.
To simplify notation we will define left and right clocks at the endpoints so that
$t_{1,r} :=t_{1}, t_{n,l} :=t_{n}$ and $\tau^{1,r}(t_{1,r}) :=\tau^{1}(t_{1}), \tau^{n,l}(t_{n,l}) :=\tau^{n}(t_{n})$.

In order to prove that both conditions of the consistency check cannot be satisifed by any set of clocks $\{\tau^{i,l}(t),\tau^{i,r}(t),i=2,\ldots,n-1\}$, we will assume that the first condition is satisfied, and show that second must fail.  Therefore, the clocks must satisfy:
\begin{align}
	\label{linearleftright}
	\tau^{i,l}(t_{i,l})&=a_{i,i-1}\tau^{i-1,r}(t_{i-1,r})+b_{i,i-1} \mbox{ for } i \leq 2 \leq n.
\end{align}
In addition, by virtue of causality, we also have:
\begin{align}
	\label{causalleftright}
	\tau^{i,l}(t_{i,l})&\leq\tau^{i,r}(t_{i,r}).
\end{align}
We prove that delay bound condition must be violated if node $1$ waits for a sufficiently large period of time before before initiating the timing packet, i.e., if $\tau^{1}(t_{1})$ is sufficiently large, then for some $i$, we have $\tau^{i,r}(t_{i,r})-\tau^{i,l}(t_{i,l})>K$.  More precisely, we show 
$\sum^{n-1}_{i=2}\left(\tau^{i,r}(t_{i,r})-\tau^{i,l}(t_{i,l})\right)>nK$,
which implies that some node has violated delay bound condition.  

The sum $\sum^{n-1}_{i=2}\left(\tau^{i,r}(t_{i,r})-\tau^{i,l}(t_{i,l})\right)$ cannot be directly evaluated because the left and right clocks $\{\tau^{i,l}(t),\tau^{i,r}(t)\}$ are arbitrary functions of $t$. However, we have the following equality by repeated addition and subtraction $\tau^{n,l}(t_{n,l}) =\tau^{1,r}(t_{1,r})+\sum^{n}_{i=2}\left(\tau^{i,l}(t_{i,l})-\tau^{i-1,r}(t_{i-1,r})\right) = \sum^{n-1}_{i=2}\left(\tau^{i,l}(t_{i,l})-\tau^{i-1,r}(t_{i-1,r})\right) =\tau^{1,r}(t_{1,r})+S_{1}+S_{2}$,
where $S_{1} :=\sum^{n}_{i=2}\left(\tau^{i,l}(t_{i,l})-\tau^{i-1,r}(t_{i-1,r})\right), S_{2} :=\sum^{n-1}_{i=2}\left(\tau^{i,l}(t_{i,l})-\tau^{i-1,r}(t_{i-1,r})\right)$.
The value $S_{2}$ is the sum of the forwarding delays.  We will use (\ref{linearleftright}) and (\ref{causalleftright}) to obtain an upper bound on $S_{1}$.  Inserting this upper bound and using the fact that $\tau^{n,l}(t)$ and $\tau^{1,r}(t)$ are both affine functions of $t$, will allow us to obtain a lower bound on $S_{2}$.  The proof will then follow easily.  
We now obtain an upper bound on $S_{1}$ when the forward skew product $\prod^{j}_{i=2}\hat{a}_{i,i-1} \geq 1$ for all $j\geq2$. 

\begin{lemma}
\label{istarmaxsumlemma}
Suppose $\prod^{j}_{i=2}a_{i,i-1}\geq 1$ for $2\leq i\leq n$.  Then
$ \sum^{n}_{i=2}(\tau^{i,l}(t_{i,l})-\tau^{i-1,r}(t_{i-1,r}))\leq\left(\frac{\hat{a}_{n,1}-1}{\hat{a}_{n,1}}\right)\tau^{n,l}(t_{n,l})
 \sum^{n}_{i=2}\frac{\hat{b}_{i,i-1}}{\hat{a}_{i,1}}$.

\end{lemma}

\begin{proof}
    We have by definition $\tau^{n+1,l}(t_{n,l}):=\hat{a}_{n+1,n}\tau^{n,r}(t_{n,r})+\hat{b}_{n+1,n}$. For $n=2$, we have $\tau^{2,l}(t_{2,l})-\tau^{1,r}(t_{1,r}) =\left(\frac{a_{2,1}-1}{a_{2,1}}\right)\tau^{2,l}(t_{2,l})+\frac{b_{2,1}}{a_{2,1}}$.
    Now assume the lemma is true for $n$.  We will show that it also holds for $n+1$:
$\sum^{n+1}_{i=2}(\tau^{i,l}(t_{i,l})-\tau^{i-1,r}(t_{i-1,r}))
		=\sum^{n}_{i=2}(\tau^{i,l}(t_{i,l})-\tau^{i-1,r}(t_{i-1,r}))+\tau^{n+1,l}(t_{n+1,l})-\tau^{n,r}(t_{n,r})
\leq\left(\frac{\hat{a}_{n,1}-1}{\hat{a}_{n,1}}\right)\tau^{n,l}(t_{n,l})+\sum^{n}_{i=2}\frac{\hat{b}_{i,i-1}}{\hat{a}_{i,1}}+\tau^{n+1,l}(t_{n+1,l})-\tau^{n,r}(t_{n,r})
\leq\left(\frac{\hat{a}_{n,1}-1}{\hat{a}_{n,1}}\right)\tau^{n,r}(t_{n,r})+\sum^{n}_{i=2}\frac{\hat{b}_{i,i-1}}{\hat{a}_{i,1}}+\tau^{n+1,l}(t_{n+1,l})-\tau^{n,r}(t_{n,r})
=\left(\frac{\hat{a}_{n,1}-1}{\hat{a}_{n,1}}\right)\left(\frac{\tau^{n+1,l}(t_{n+1,l})-\hat{b}_{n+1,n}}{\hat{a}_{n+1,n}}\right)+\sum^{n}_{i=2}\frac{\hat{b}_{i,i-1}}{\hat{a}_{i,1}}
+\left(\frac{\hat{a}_{n+1,n}-1}{\hat{a}_{n+1,n}}\right)\hat{\tau}^{n+1,l}(t_{n+1,l})+\frac{\hat{b}_{n+1,n}}{\hat{a}_{n+1,n}}
 =\left(\frac{\hat{a}_{n+1,1}-1}{\hat{a}_{n+1,1}}\right)\tau^{n+1,l}(t_{n+1,l})+\sum^{n+1}_{i=2}\frac{\hat{b}_{i,i-1}}{\hat{a}_{i,1}}$,
which follow from the induction hypothesis above in the Lemma statement, and
the fact that $\tau^{n,r}(t_{n,r})\geq\tau^{n,l}(t_{n,l})$ and $a_{i,1}\geq1$ for all $2\leq i\leq n+1$ (that is, the coefficient $\hat{a}_{i,1}-1$ is negative).
\end{proof}

We next bound $S_{1}$
in the special case 
when the reverse skew product $\prod^{j}_{i=1}\hat{a}_{n-(i-1),n-i} \leq 1$ for all $j\geq1$.

\begin{lemma}
    \label{istarminsumlemma}
    Suppose $\prod^{j}_{i=1}a_{n-(i-1),n-i}\leq1$ for $2\leq j\leq n-1$.  Then
$ \sum^{j}_{i=1}(\tau^{n-(i-1),l}(t_{n-(i-1),l})-\tau^{n-i,r}(t_{n-i,r}))$
$ \leq\left(\hat{a}_{n,n-j}-1\right)\tau^{n-j,r}(t_{n-j,r})+\hat{b}_{n,n-1}+\sum^{n-1}_{i=n-j+1}\hat{a}_{n,i}\hat{b}_{i,i-1}$.
\end{lemma}

\begin{proof}
    We have by definition $\tau^{n-(k-1),l}(t_{n-(k-1),l}):=\hat{a}_{n-(k-1),n-k}\tau^{n-k,r}(t_{n-k,r})+\hat{b}_{n-(k-1),n-k}$.  For $j=1$,
$ \tau^{n,l}(t_{n,l})-\tau^{n-1,r}(t_{n-1,r})=(a_{n,n-1}-1)\tau^{n-1,r}(t_{n-1,r})$.

    Now assume the Lemma holds for $j$.  We will show that it must hold for $j+1$:
$\sum^{j+1}_{k=1}(\tau^{n-(k-1),l}(t_{n-(k-1),l})-\tau^{n-k,r}(t_{n-k,r}))
=\sum^{j}_{k=1}(\tau^{n-(k-1),l}(t_{n-(k-1),l})-\tau^{n-k,r}(t_{n-k,r}))
+\tau^{n-j,l}(t_{n-j,l})-\tau^{n-(j+1),r}(t_{n-(j+1),r})
\leq(\hat{a}_{n,n-j}-1)\tau^{n-j,r}(t_{n-j,r})+\hat{b}_{n,n-1}
 +\sum^{n-1}_{k=n-j+1}\hat{a}_{n,k}\hat{b}_{k,k-1}
+\tau^{n-j,l}(t_{n-j,l})-\tau^{n-(j+1),r}(t_{n-(j+1),r})
\leq(\hat{a}_{n,n-j}-1)\tau^{n-j,l}(t_{n-j,l})+\hat{b}_{n,n-1}
 +\sum^{n-1}_{k=n-j+1}\hat{a}_{n,k}\hat{b}_{k,k-1}
 +\tau^{n-j,l}(t_{n-j,l})-\tau^{n-(j+1),r}(t_{n-(j+1),r})
 \leq(\hat{a}_{n,n-(j+1)}-1)\tau^{n-(j+1),r}(t_{n-(j+1),r})+\hat{b}_{n,n-1}
 +\sum^{n-1}_{k=n-j}\hat{a}_{n,k}\hat{b}_{k,k-1}$.
The above follow from induction hypothesis in Lemma \ref{istarminsumlemma}, 
since $\tau^{i,l}(t_{i,l})\leq\tau^{i,r}(t_{i,r})$ and $\hat{a}_{n,n-j}\leq1$ for $1\leq j\leq n-1$ (that is, the coefficient $\hat{a}_{n,n-j}-1$ is negative), and from substitution into $\tau^{n-j,l}(t_{n-j,l})$ and simplification.  
    \end{proof}

    We will combine both special cases in Lemma \ref{istarmaxsumlemma} and Lemma \ref{istarminsumlemma} to obtain an upper bound on $S_{1}$.  First we define $i^{*}$ as the node with the smallest skew product $\hat{a}_{i^{*},1}$ in the chain network, that is less than one.  That is, $\hat{a}_{i^{*},1}=\displaystyle\min_{k}\hat{a}_{k,1}$ and $\hat{a}_{i^{*},1}\leq1$.  If no such node exists, set $i^{*}=1$. 

  Now we consider an arbitrary set of skews $\{\hat{a}_{i,i-1},i=2,\ldots,n\}$.  Next we show that if $i^{*}\geq2$ then the forward skew product starting from $i^{*}$ is greater than 1, and the reverse skew product starting from $i^{*}-1$ is always less than one.

\begin{lemma}
\label{lemmaistar}
If $i^{*}\geq2$ then $\hat{a}_{j,i^{*}}\geq1$ for $i^{*}+1\leq j\leq n$ and $\hat{a}_{i^{*},i^{*}-k+1}\leq1$ for $1\leq k\leq i^{*}$.  Otherwise, $\hat{a}_{j,1}\geq1$ for $2\leq j\leq n$.
\end{lemma}

\begin{proof}
	Consider $i^{*}\geq2$, and suppose the first part of the assertion is false.  I.e., for some $j^{\prime}$, $\hat{a}_{j^{\prime}i^{*}}<1$.  It follows that $\hat{a}_{j^{\prime}1}=\hat{a}_{j^{\prime}i^{*}}\hat{a}_{i^{*}1}\leq\hat{a}_{i^{*}1}$.  But then $j^{\prime}$ is a node with a smaller skew product $\hat{a}_{j1}$ than node $i^{*}$, which contradicts the definition of $i^{*}$.  Now suppose that the second part of the assertion is false.  I.e., for some $j^{\prime}$ we have $\hat{a}_{i^{*}j^{\prime}}>1$.  It follows that $\hat{a}_{i^{*}1}=\hat{a}_{i^{*}j^{\prime}}\hat{a}_{j^{\prime}1}\geq\hat{a}_{j^{\prime}1}$.  But then $j^{\prime}$ is a node with a smaller skew product than node $i^{*}$, which again contradicts the definition of $i^{*}$.  Now consider the case when $i^{*}=1$.  Then by definition of $i^{*}$ it follows that $\hat{a}_{j1}\geq1$ for all $2\leq j\leq n$.
\end{proof}

We now obtain an upper bound on $S_{1}$ for arbitrary skews.
\begin{lemma} \label{allskewlemma}
Suppose $i^{*}\geq2$.  We have the following inequality:
$\sum^{n}_{j=2}\tau^{j,l}(t_{j,l})-\tau^{j-1,r}(t_{j-1,r})\leq(\hat{a}_{i^{*},1}-1)\tau^{1,r}(t_{1,r})$
$ +\left(\frac{\hat{a}_{n,i^{*}}-1}{\hat{a}_{n,i^{*}}}\right)\tau^{n,l}(t_{n,l})+\frac{\hat{b}_{n,1}}{\hat{a}_{n,i^{*}}}$.
\end{lemma}

\begin{proof}
$\sum^{n}_{j=2}\tau^{j,l}(t_{j,l})-\tau^{j-1,r}(t_{j-1,r})$
$ \sum^{i^{*}}_{j=2}\tau^{j,l}(t_{j,l})-\tau^{j-1,r}(t_{j-1,r})+\sum^{n}_{j=i^{*}+1}\tau^{j,l}(t_{j,l})-\tau^{j-1,r}(t_{j-1,r})$
$ =(\hat{a}_{i^{*},1}-1)\tau^{1,r}(t_{1,r})+\hat{b}_{i^{*},i^{*}-1}$
$ +\sum^{i^{*}}_{j=2}\hat{a}_{i^{*},j}\hat{b}_{j,j-1}+\left(\frac{\hat{a}_{n,i^{*}}-1}{\hat{a}_{n,i^{*}}}\right)\tau^{n,l}(t_{n,l})+\sum^{n}_{i=i^{*}+1}\frac{\hat{b}_{i,i-1}}{\hat{a}_{i,i^{*}}}$
$ =(\hat{a}_{i^{*},1}-1)\tau^{1,r}(t_{1,r})+\left(\frac{\hat{a}_{n,i^{*}}-1}{\hat{a}_{n,i^{*}}}\right)\tau^{n,l}(t_{n,l})$
$ +\sum^{n}_{j=2}\frac{\hat{a}_{n,j}\hat{b}_{j,j-1}}{\hat{a}_{n,i^{*}}}$
$ =(\hat{a}_{i^{*},1}-1)\tau^{1,r}(t_{1,r})+\left(\frac{\hat{a}_{n,i^{*}}-1}{\hat{a}_{n,i^{*}}}\right)\tau^{n,l}(t_{n,l})+\frac{\hat{b}_{n,1}}{\hat{a}_{n,i^{*}}}.
$
which follow by applying Lemma \ref{istarminsumlemma} and Lemma \ref{istarmaxsumlemma}, by multiplying the terms in each summation by $\frac{\hat{a}_{n,i^{*}}}{\hat{a}_{n,i^{*}}}$ and simplifying, and from the definitions of $\hat{b}_{ij}$ and $\hat{d}^{(i)}_{ji}$.
\end{proof}

Now that we have an upper bound on $S_{1}$, we can
obtain a lower bound on $S_{2}$, the sum of the forwarding delays.  

\begin{lemma} \label{delaysumlemma}
The sum of forwarding delays in the chain network satisfies:
$\sum^{n-1}_{j=2}\left(\tau^{j,l}(t_{j,l})-\tau^{j,r}(t_{j,r})\right) \geq\frac{(a_{n,1}-\hat{a}_{n,1})}{\hat{a}_{n,i^{*}}}\tau^{1,r}(t_{1,r})+\frac{(b_{n,1}-\hat{b}_{n,1})}{\hat{a}_{n,i^{*}}}$.
\end{lemma}

\begin{proof}
$\sum^{n-1}_{j=2}\left(\tau^{j,l}(t_{j,l})-\tau^{j,r}(t_{j,r})\right)$
$ =\tau^{n,l}(t_{n,l})-\tau^{n,r}(t_{n,r})-\sum^{n}_{j=2}\tau^{j,l}(t_{j,l})-\tau^{j-1,r}(t_{j-1,r})$
$ \geq\tau^{n,l}(t_{n,l})-\tau^{n,r}(t_{n,r})-(\hat{a}_{i^{*},1}-1)\tau^{1,r}(t_{1,r})$
$ -\left(\frac{\hat{a}_{n,i^{*}}-1}{\hat{a}_{n,i^{*}}}\right)\tau^{n,l}(t_{n,l})-\frac{\hat{b}_{n,1}}{\hat{a}_{n,i^{*}}}$
$ =\frac{\tau^{n,l}(t_{n,l})}{\hat{a}_{n,i^{*}}}-\hat{a}_{i^{*},1}\tau^{1,r}(t_{1,r})-\frac{\hat{b}_{n,1}}{\hat{a}_{n,i^{*}}}$
$ \geq\frac{\tau^{n,l}(t_{1,r})}{\hat{a}_{n,i^{*}}}-\hat{a}_{i^{*},1}\tau^{1,r}(t_{1,r})-\frac{\hat{b}_{n,1}}{\hat{a}_{n,i^{*}}}$
$ =\frac{a_{n,1}\tau^{1,r}(t_{1,r})+b_{n,1}}{\hat{a}_{n,i^{*}}}-\hat{a}_{i^{*},1}\tau^{1,r}(t_{1,r})-\frac{\hat{b}_{n,1}}{\hat{a}_{n,i^{*}}}$
$ =\frac{(a_{n,1}-\hat{a}_{n,1})}{\hat{a}_{n,i^{*}}}\tau^{1,r}(t_{1,r})+\frac{(b_{n,1}-\hat{b}_{n,1})}{\hat{a}_{n,i^{*}}}$,
which follow by noting from repeated addition and subtraction that $\tau^{n,l}(t_{n,l})=\tau^{1,r}(t_{1,r})+\sum^{n}_{j=2}\left(\tau^{j,l}(t_{j,l})-\tau^{j-1,r}(t_{j-1,r})\right)+\sum^{n-1}_{j=2}\left(\tau^{j,l}(t_{j,l})-\tau^{j,r}(t_{j,r})\right)$, by applying Lemma \ref{allskewlemma}, because $t_{n,l}\geq t_{1,r}$ since node $n$ could not have received the timing packet before node $1$ transmitted it, and since node $n$'s clock is relatively affine with respect to node $1$'s clock.
\end{proof}

We now complete the proof of consistency check for a chain network.  We show that if the start time of the consistency check is sufficiently large, and the left and right clocks $\{\tau^{i,l}(t_{i,l}),\tau^{i,r}(t_{i,r})\}$ satisfy the parameter consistency condition, then at least one node will violate delay bound condition.  Hence there are no left and right clocks that can pass both conditions of consistency check if start time is large. 

\begin{proof}
	We assume node $1$ is a good node.  Now 
$\frac{(a_{n,1}-\hat{a}_{n,1})}{\hat{a}_{n,i^{*}}}\tau^{1,r}(t_{1,r})+\frac{(b_{n,1}-\hat{b}_{n,1})}{\hat{a}_{n,i^{*}}}>nK$.
But by Lemma \ref{delaysumlemma} the LHS of this inequality is the lower bound of the sum of the delays in the chain $\sum^{n}_{j=2}\left(\tau^{j,l}(t_{j,l})-\tau^{j,r}(t_{j,r})\right)$.  By substitution, 
$\sum^{n}_{j=2}\left(\tau^{j,l}(t_{j,l})-\tau^{j,r}(t_{j,r})\right)>nK$.
It follows that for some malicious node $j\in\{2,\ldots,n\}$, $\tau^{j,l}(t_{j,l})-\tau^{j,r}(t_{j,r})>K$ which violates the delay bound condition.
\end{proof}

    Now we can show that neighbor and network discovery phases together allow the good nodes to form a rudimentary network, where the good nodes have the same topological view and consistent estimates of a reference clock.  The first obstacle is that the protocol is composed of stages that must be completed sequentially by all the nodes in the network, even prior to clock synchronization.  
    Suppose that $[t_{k},t_{k+1})$ is the interval allocated to the $kth$ stage.  Any messages transmitted between adjacent good nodes must arrive in the same interval they were transmitted. Since send-times are measured with respect to the source clock, and  receive-times with respect to the destination clock,
      the intervals must be chosen large enough to compensate for the maximum clock divergence caused by skew $a_{ij}\leq a_{max}$ and offset $b_{ij}\leq a_{max}U_{0}$.  

    \begin{lemma}
	    \label{timeintervals}
	    There exists a sequence of adjacent time-intervals $[t_{k},t_{k+1})$ and corresponding schedule that guarantees any message of size $W$ transmitted (via OMC) by node $i$ in the interval $[t_{k},t_{k+1})$ (as measured by $i$'s clock) will be received by node $j$ in the same interval as measured by node $j$'s clock.
    \end{lemma}
    \begin{proof}
	    Set $t_{k+1}:=(a_{max})^{2}t_{k}+2(a_{max})^{3}U_{0}\\*+(a_{max})^{3}T_{MAC}(W)$.
	     Suppose a message from node $i$ to node $j$ during $[t_{k},t_{k+1})$ is transmitted (via the OMC) at $t_{s}:=a_{max}t_{k}+(a_{max})^{2}U_{0}$ with respect to node $i$'s clock.  
	     By substitution and simplification it follows that $\tau^{j}_{i}(t_{s})\geq t_{k}$ and
	     $\tau^{j}_{i}(t_{s}+T_{MAC}(W))<t_{k+1}$.
	     Hence $\tau^{j}_{i}([t_{s},t_{s}+T_{MAC}(W)))\subset[t_{k},t_{k+1})$, and
	     so $j$ receives this message during the same interval with respect to $j$'s clock.        
    \end{proof}



    \begin{theorem}
	    \label{networkdiscovery}
	    After Network Discovery, the good nodes have a common view of the topology and consistent estimates (to within $\epsilon_{a}$) of the skew of the reference clock.
    \end{theorem}

    \begin{proof}
	    From Lemma \ref{timeintervals} 
	    all good nodes 
	    will proceed through each stage of Neighbor and Network Discovery Phases together, and therefore 
	    establish link certificates with their good neighbors.  Since they form a connected component, the good nodes obtain a common view of their link certificates using the EIGByzMAC algorithm and the schedule in Lemma \ref{timeintervals}.  The good nodes can therefore infer the network topology and the relative skews of all adjacent nodes based upon the collection of link certificates.  Using Consistency Check, the good nodes can eliminate paths along which bad nodes have provided false skew data.  The good nodes can disseminate this information to each other using the EIGByzMAC algorithm and Lemma \ref{timeintervals} and thus obtain consistent estimates of the reference clock to within $\epsilon_{a}$.     
    \end{proof}


    \begin{lemma}
	    The sequence of adjacent intervals $[t_{j},t_{j+1})$, $j=0,\ldots,k$ is contained in $[t_{0}, c_{1}t_{0}+c_{2}W)$ where constants
	    $c_{1}$ and $c_{2}$ depend on $a_{max}$, $k$, $U_{0}$, and $n$.
    \end{lemma}
    
    \begin{proof}
	   For the OMC $T_{MAC}(W)\leq cW$, where $c$ depends on $a_{max}$, and $n$.  The result for $k=1$ follows from definition of $t_{k}$, and substitution of $cW$ into $T_{MAC}(W)$, and for general $k$ by induction and definition of $t_{k}$.
    \end{proof}

        \begin{lemma}
		\label{networkdiscoverytime}
	    The time to complete Neighbor and Network Discovery Phases $T_{nei}+T_{net}$ is less than $c_{1}\log T_{life}+\frac{c_{2}}{\epsilon_{a}}$ where $c_{1},c_{2}$ depend only on $n,a_{max},
	    U_{0}$.
    \end{lemma}

    \begin{proof}
	    From Algorithms \ref{NeighborDiscoveryAlg}, \ref{EIGByzMAC}, \ref{ConsistencyCheck} and \ref{NetworkDiscovery} there are at most $6+n+n|C|+n$ protocol stages in the Neighbor and Network Discovery Phases. Hence the time required is at most $c_{1}t_{0}+c_{2}W$, where $W$ is the size of a message to be transmitted, and $c_{1},c_{2}$ are constants depending on the number of protocol stages $a_{max},U_{0},n$.  The maximum size of a message is proportional to the timing packet size $\log T_{life}$.  To account for the effect of the minimum start-time $T_{s}$ for the consistency check, we can assume the worst case that the $T_{s}$ comes into effect during the first protocol stage (instead of later in the Network Discovery Phase).  From Theorem \ref{bigtheorem} the consistency check start-time is at most $\frac{c}{\epsilon_{a}}$, where $c$ depends on $U_{0}, a_{max},n$.  Substitution into $t_{0}$ proves the lemma.
    \end{proof}

    \begin{lemma}
	    The time required for the Data Transfer Phase is at most $c_{3}B+c_{4}D$ where $B$ is the time spent transmitting data packets, $D$ is the size of the dead-time separating  time slots, and $c_{3},c_{4}$ depend on $n$ alone.
    \end{lemma}

    \begin{proof}
	    The total number of time-slots for data transfer between all source-destination pairs is $n^{2}(n-1)$,  each supporting data transfer of size $B_{s}$ and a dead-time $D$.
    \end{proof}

    \begin{lemma}
	    The time required for the Verification Phase is at most $c_{5}D$ where $c_{5}$ depends on $n$ alone.
    \end{lemma}

    \begin{proof}
	    In each stage of the EIG Byzantine General's algorithm, there are at most $n!$ vertex values that must be transmitted with each node in the neighborhood.  The value of a vertex is a list of CTVs.  There are at most $2^{n}$ CTVs and at most $n$ nodes in a CTV.  Therefore the size of any message to be transmitted by a node during EIG algorithm is at most $cD$, where $c$ is a constant dependent on $n$.  Since there are $n(n-1)$ possible source-destination pairs, there are at most $n(n-1)$ time slots in each stage, separated at the beginning and end by a dead-time $D$.  Therefore the duration of each stage is at most $cD+n(n-1)2D$.  There are at most $n$ stages.  
    \end{proof}

    We can now prove the main theorem of this paper.

    \begin{theorem} \label{main}
	    The protocol ensures that the network proceeds from startup to a functioning network carrying data.  There exists a selection of parameters $n_{iter}$, $D$, $B$, $\epsilon_{a}$ and $T_{life}$ that achieves min-max utility
	    over the enabled set, to within a factor $\epsilon$, where the min is over all policies of the bad nodes that  can only
	    adopt two actions in each CTV: conform to the protocol and/or jam.  The achieved utility is $\epsilon$-optimal.  
    \end{theorem}

    \begin{proof}
	We begin by choosing parameters
	so that the protocol overhead, which includes Neighbor Discovery, Network Discovery, Verification, all dead-times, and iterations converging to the final rate vector, is an arbitrarily small fraction of the total operating lifetime. 
	With $\hat{\tau}^{r}_{i}(t):=\hat{a}_{ri}t$ the estimate of reference clock $r$ with respect to the local clock at node $i$, the maximum difference in nodal estimates is bounded as
	$
		|\hat{\tau}^{r}_{i}(\tau^{i}(t))-\hat{\tau}^{r}_{k}(\tau^{k}_{i}(\tau^{i}(t)))|\leq2(a_{max})^{2}\epsilon_{a}T_{life}+(a_{max})^{2}U_{0}
	$.
	With $k_{r}$ be the number of rate vectors in the rate region, we can
	   choose $n_{iter}$, $D$, $B$, $\epsilon_{a}$ and $T_{life}$ to satisfy:
	   $
      	\frac{n_{iter}}{n_{iter}+2^{n}k_{r}}\geq1-\epsilon_{l}
   $,
   $
    	\frac{B}{c_{1}\log T_{life}+\frac{c_{2}}{\epsilon_{a}}+B+c_{3}D+c_{4}D}\geq1-\epsilon_{d}
    	$,
	$
    	n_{iter}((c_{1}\log T_{life}+\frac{c_{2}}{\epsilon_{a}}+B+c_{3}D+c_{4}D)\leq T_{life}
    	$,
	$
    	2(a_{max})^{2}\epsilon_{a}T_{life}+(a_{max})^{2}U_{0}\leq D
    $.
    These ensure
    that the rate loss due to failed CTVs 
is arbitrarily small,
the time spent transmitting data is an arbitrarily large fraction of the duration of that iteration,  
the operating lifetime is large enough to support $n_{iter}$ protocol iterations, 
and the dead-time $D$ is large enough to tolerate the maximum divergence in clock estimates caused by skew error $\epsilon_{a}$.  
 	
	Let $\{ {\cal{D}}(t) \}$ be the decreasing sequence of sets of
	disabled CTVs, with 
	limit $\bar{\cal{D}}$ attained at some finite time $T$.
	Suppose 
	$x$ achieves 
the maximum utility
for $\bar{\cal{D}}$
over the nodes in the same component as the good nodes.
No protocol can do better when $\bar{\cal{D}}$ is disabled.
	 The proposed protocol attains $x(1-\epsilon_{d})(1-\epsilon_{l})$.  
    \end{proof}
%
%

\section{Concluding Remarks} \label{Concluding Remarks}
We have presented a complete suite of protocols that enables a collection of good nodes interspersed with bad nodes to form a functioning network from start-up, operating at a utility-optimal rate vector, regardless of what the bad nodes conspire to do,
    under a certain system model.  
    Further, the attackers cannot decrease the utility any more than they could by just conforming to the protocol or jamming on each CTV.

This paper is only an initial attempt to obtain a theoretical foundation for a much needed holistic all-layer approach to secure wireless networking, and there are several open issues. An important potential generalization is to allow probabilistic communication.  
%
Since the protocol presented has poor transient behavior, though overall optimal,
it
needs to be explored how to increase efficiency 
in the transient phase.

Much further work remains to be done.   



%

%


\ifCLASSOPTIONcaptionsoff
  \newpage
\fi



%


\bibliographystyle{abbrv}
\noindent

%
%
%
%
%
%
%
%
%
\end{document}